\documentclass[11pt,letterpaper]{article}
\usepackage[margin=1in]{geometry}
\usepackage{comment}
\usepackage[backref=page]{hyperref}
\hypersetup{
  colorlinks=true,
  citecolor=blue
}
\usepackage{graphicx}
\usepackage{mathtools}
\usepackage{amsmath}
\usepackage{amsthm}
\usepackage{amssymb}
\usepackage[ruled,vlined]{algorithm2e}
\usepackage{xcolor}
\newcommand{\cmt}[1]{}

\newtheorem{theorem}{Theorem}[section]

\newtheorem{proposition}[theorem]{Proposition}
\theoremstyle{remark}
\newtheorem*{remark}{Remark}

\title{Improved Approximation Ratios of Fixed-Price Mechanisms in Bilateral Trades}
\usepackage{authblk}
\author[1]{Zhengyang Liu\thanks{zhengyang@bit.edu.cn}}
\author[2]{Zeyu Ren\thanks{zeyuren@ruc.edu.cn}}
\author[2]{Zihe Wang\thanks{wang.zihe@ruc.edu.cn}}
\affil[1]{Beijing Institute of Technology}
\affil[2]{Renmin University of China}
\date{}

\begin{document}
\maketitle
\thispagestyle{empty}
\begin{abstract}
We continue the study of the performance for fixed-price mechanisms in the bilateral trade problem, and improve approximation ratios of welfare-optimal mechanisms in several settings. Specifically, in the case where only the buyer distribution is known, we prove that there exists a distribution over different fixed-price mechanisms, such that the approximation ratio lies within the interval of $[0.71, 0.7381]$. Furthermore, we show that the same approximation ratio holds for the optimal fixed-price mechanism, when both buyer and seller distributions are known. As a result, the previously best-known $(1 - 1/e+0.0001)$-approximation can be improved to $0.71$. Additionally, we examine randomized fixed-price mechanisms when we receive just one single sample from the seller distribution, for both symmetric and asymmetric settings. Our findings reveal that posting the single sample as the price remains optimal among all randomized fixed-price mechanisms.
\end{abstract}

\section{Introduction}
\label{sec:introduction}

In this paper, we study the bilateral trade problem in a basic
model: one buyer and one seller want to trade a single and indivisible
good. Both of them have their own private valuations to the good, drawn from  public distributions respectively. The mechanism designer can access both of their
value distributions, but not their particular values. Depending on whether value
distributions of these two agents are identical, we have the {\em symmetric}
setting and the {\em asymmetric} (or general) setting. The designer aims to maximize the {\em social welfare}, i.e., the value of the one
holding the good after the trade (if any).

Formally, we denote by $B$ the value of the buyer, and $S$ the value of the seller. The ideal solution is the so-called ``first-best'' mechanism, which is the case that
the trade happens once $B>S$ and the corresponding social welfare would be $\max \left\{ B,S \right\}$. The celebrated result due to~\cite{Myerson:1983uq} states that no
individually rational (IR) and Bayesian incentive-compatible (BIC) mechanism
with the budget-balanced (BB) constraint can achieve the first-best optimum. In the
same paper, Myerson and Satterthwaite also present a ``second-best'' mechanism
attaining the optimal among all the IR, BIC and BB ones. However, a drawback of the second-best mechanism is that the characterization is
complicated to describe and thus not suitable for practical implementation. Hence we need to pursue simple and almost-optimal mechanisms. 

Fixed-price mechanisms are the best choice in the literature~\cite{Hagerty:1987vj,mcafee2008gains,colini2017fixed,brustle2017approximating}.  In fixed-price mechanisms, the mechanism designer
picks a price $p$ (depends on the distributions of two agents, possibly) and gives it to both
of the buyer and the seller. In the {\em randomized} case, one can pick the price drawn from some distribution. If the price is between the values of two agents (i.e., $S\le p\le B$),
the trade happens. It is easy to see that fixed-price mechanisms are
dominant-strategy incentive-compatible (DSIC). A recent research line tries to prove
the near-optimal guarantees of such mechanisms. Our paper improves approximation ratios in variant settings for this problem.

\subsection{Related Work}
\label{sec:related-work}
Forty years ago, Myerson and Satterthwaite~\cite{Myerson:1983uq} initiated the
study of the optimal incentive-compatible mechanisms in bilateral trades. In
\cite{Hagerty:1987vj}, Hagerty and Rogerson showed that fixed-price mechanisms are
the unique choice of strong budget balanced and dominant-strategy truthful
mechanisms in bilateral trades, essentially. By relaxing the budget balance
condition to a no-deficit condition, fixed-price mechanisms can be optimal due
to~\cite{Drexl:2015tn,Shao:2016vw}.

Recently, a research line pursued near-optimal guarantees of fixed-price
mechanisms in bilateral trades w.r.t. the welfare. Blumrosen and Dobzinski~\cite{Blumrosen:2014tr} considered the Median Mechanism which sets a price equals to the median of the distribution of the seller and showed that the approximation ratio is $0.5$. Blumrosen and Dobzinski~\cite{Blumrosen:2021uz} also showed that given the distribution of either the seller or the buyer, there is a price distribution such that the randomized fixed-price mechanism in expectation achieves at least $1-1/e\approx 0.63$ of the optimal efficiency. It implies that for every pair of the seller and the buyer distributions, there exists a deterministic fixed-price mechanism that achieves at least $1-1/e\approx 0.63$ of the optimal efficiency. Recently, Kang, Pernice and Vondr\'{a}k~\cite{Kang:2022wo} showed that in the general setting, the approximation ratio can be greater than $1-1/e+0.0001$, beating the previous work. Their method is choosing the better one between two mechanisms depending on whether the asymmetry between two agents' distributions is severe or not. In the same paper, they also provided the tight ratio $(2+\sqrt{2})/4$ for the symmetric setting. On the other side, \cite{Colini-Baldeschi:2016ut} gave an example of the seller and the buyer distributions where no fixed-price mechanism can achieve $0.7485$ of the optimal efficiency. This impossibility result is improved to $0.7385$ by~\cite{Kang2018StrategyProofAO}. 

Another setting captures the case where we can only obtain the minimum-possible amount
of the prior information from the distributions. For the general setting, \cite{Dutting:2021wx} provided a $1/2$-approximation analysis
by using the single sample from the seller distribution as the price. For the symmetric setting, Kang, Pernice and
Vondr\'{a}k~\cite{Kang:2022wo} showed that a $3/4$-approximation to the welfare
can be obtained.

Another notion of characterizing the performance of mechanisms in bilateral
trades is {\em gain from trade} (GFT), i.e., the expected social utility
attained from the trading. Very recently, \cite{Deng:2022uy} and \cite{Fei:2022uh} proved that a combination of simple mechanisms (RandOff) achieves the constant
approximation of the first-best, and the state-of-the-art ratio is $1/3.15$. Obtaining the exact ratio is a well-known open problem.
On the hardness side, \cite{Babaioff:2021wj} and \cite{Cai:2021tp} independently showed that there exist instances that RandOff cannot achieve half to the optimal GFT.
Besides, \cite{Leininger:1989tq} and \cite{Blumrosen:2016uv} proved that no BIC, IR mechanisms can obtain better approximation than $2/e$ to the optimal GFT.
\paragraph{Independent Work.}
Very recently, Cai and Wu~\cite{cai2023optimal} found that the infinite dimensional quadratically constrained quadratic program (QCQP) can characterize the worst-case instance and the approximation ratio is determined by such an instance. They used a finite program to get a lower bound of $0.72$ and a upper bound of $0.7381$ numerically, respectively. It is worth noting that solving the general case of QCQP is an NP-hard problem. Compared with their method, our algorithm is FPTAS. Thus, to approach the tight bound numerically, our method may be more efficient.
\subsection{Our Results}
\label{sec:our-results}
In this paper, we reduce gaps mentioned in~\cite{Kang:2022wo} as shown in Table~\ref{tab:summary}:
\begin{enumerate}
\item For any distribution of the buyer, there is a distribution of prices such that the randomized fixed-price mechanism in expectation achieves $0.71$ approximation of the optimal efficiency. The main tools we use are interesting characterizations of this price distribution (See Section~\ref{sec:proper}). Moreover, we characterize the extreme distribution of the buyer where the tight bound can be achieved. We also provide an example where no randomized fixed-price mechanism can achieve $0.7381$ approximation. We show that the same approximation ratio holds for the asymmetric setting. Thus, a direct corollary is that for every pair of buyer and seller distributions, there is a fixed-price mechanism that achieves $0.71$ approximation of the optimal efficiency, improving the previous $(1 - 1/e+0.0001)$-approximation.

\item For the general bilateral trade using a single sample from the seller distribution as the sole information, we show that even for the randomized fixed-price mechanisms, it is impossible to achieve $1/2$-approximation. Together with the previous work~\cite{Dutting:2021wx}, the approximation ratio $1/2$ is tight in this setting.
In addition, we show that any randomized fixed-price mechanism using a sample as information cannot attain an approximation better than $3/4$.
in the symmetric setting. Together with the previous work~\cite{Kang:2022wo}, the ratio $3/4$ is also tight.
\end{enumerate}
\begin{table}[!htbp]\centering
\vspace{-0.5cm}
\begin{tabular}{|c|c|}\hline
  {\bf Variant}& {\bf Approximation} \\ \hline
  Asymmetric, full knowledge & $\textcolor{magenta}{[0.71,0.7381]}$  \\ \hline
  Asymmetric, buyer's knowledge & $\textcolor{magenta}{[0.71,0.7381]}$   \\ \hline
  Symmetric, full knowledge & $(2+\sqrt{2})/4^{*}$\cite{Kang:2022wo} \\ \hline
  Asymmetric, 1 prior sample from seller & $\textcolor{magenta}{1/2}^{*}$ (also see~\cite{Dutting:2021wx})  \\ \hline
  Symmetric, 1 prior sample  & $\textcolor{magenta}{3/4}^{*}$ (also see~\cite{Kang:2022wo})  \\ \hline
\end{tabular}
\caption{A Survey of fixed-price mechanism approximations in bilateral trades
  w.r.t welfare maximization. Following~\cite{Kang:2022wo}, we also highlight our
  results in \textcolor{magenta}{magenta}, and mark a $*$ for all the tight results. For cases with 1 prior sample, we get tight results by studying randomized mechanisms.}
\label{tab:summary}
\vspace{-0.5cm}
\end{table}
\section{Preliminaries}
\label{sec:preliminaries}
In the bilateral trade, a buyer and a seller trade an indivisible good. They
have independent values to the good, $b$ for the buyer and $s$ for the
seller. The specific value are unknown to the mechanism designer. We assume that
their values are drawn from public distributions $F_B$ and $F_S$, i.e., $b\sim F_B$ and $s\sim
F_S$. The expected gain from trade achieved by using fixed price $p$ can be represented by
\begin{align*}
   &\mathbb{E}_{B,S}[(b-s)\cdot\mathbf{1}_{b\geq p \geq s}] \\
    =&\mathbb{E}_{B,S}[(b-p+p-s)\cdot\mathbf{1}_{b\geq p \geq s}]\\
    =&\Pr(s\le p)\mathbb{E}_{B}[(b-p)\cdot\mathbf{1}_{b\geq p}]+\Pr(b\ge p)\mathbb{E}_S[(p-s)\cdot\mathbf{1}_{p\ge s}]\\
    \ge&F_S(p)\mathbb{E}_{B}[(b-p)\cdot\mathbf{1}_{b\geq p}]+ (1-F_B(p))\mathbb{E}_S[(p-s)\cdot\mathbf{1}_{p\ge s}] \\
    =&F_S(p)\int_p^{\infty}\left(1-F_B(b)\right)\mathrm{d}b+(1-F_B(p))\mathbb{E}_S[(p-s)\cdot\mathbf{1}_{p\ge s}]\\
    =&\mathbb{E}_S\left[\left(\int_p^{\infty}(1-F_B(b))\mathrm{d}b+(1-F_B(p))(p-s)\right)\cdot\mathbf{1}_{p\ge s}\right].
\end{align*}

The inequality is due to that we ignore the case where the buyer's value is exactly the price. When the buyer distribution is continuous, the inequality becomes  an equality. The welfare can be separated into the sum of the seller's value of the good and the gain from trade. Thus, We define the expected welfare $\textrm{W}(p;F_S,F_B):=\mathbb{E}_{S}[s]+\mathbb{E}_{B,S}[(b-s)\cdot\mathbf{1}_{b\geq p \geq s}]$.  We also let $W_1(p;F_S,F_B)$ denote the expression:
$$\mathbb{E}_S\left[s+\left(\int_p^{\infty}(1-F_B(b))\mathrm{d}b+(1-F_B(p))(p-s)\right)\cdot\mathbf{1}_{p\ge s}\right].$$

We denote by $\textrm{OPT-W}(F_S,F_B)$ the first-best optimal welfare, that is $\max\{B, S\}$. We can also rewrite it as:
\begin{align*}
    \textrm{OPT-W}(F_S,F_B)&=\mathbb{E}_S[s]+\mathbb{E}_{B,S}[(b-s)\cdot\mathbf{1}_{b\geq s}] \\
    &=\mathbb{E}_S[s]+\mathbb{E}_S[\mathbb{E}_B[(b-s)\cdot \mathbf{1}_{b\ge s} ]]\\
    &=\mathbb{E}_S\left[s+\int^{\infty}_s (1-F_B(b))\mathrm{d}b\right].
\end{align*}

If the distributions $F_S$ and $F_B$ are
identical, we use $F$ instead and rewrite $\textrm{W}(p;F_S,F_B)$ and $\textrm{OPT-W}(F_S,F_B)$ as $\textrm{W}(p;F)$ and
$\textrm{OPT-W}(F)$, respectively.

In this paper, we consider the welfare approximation ratios of the optimal fixed-price mechanism,
\begin{equation*}
    \inf_{F_B,F_S\in\Delta_{L^1}(\mathbb{R}_{+})}\sup_{p\in\mathbb{R}_{+}}\frac{\mathrm{W}(p;F_S,F_B)}{\mathrm{OPT}\text{-}\mathrm{W}(F_S,F_B)},
\end{equation*}
where $\Delta_{L^1}(\mathbb{R}_+)$ denotes the set of all the probability distributions over non-negative real numbers. Both $E_B[b]$ and $E_S[s]$ should be greater than 0.
We are interested in the case where $\mathrm{OPT}\text{-}\mathrm{W}(F_S,F_B)$ is finite and it requires both $E_B[b]$ and $E_S[s]$ are finite.
For any function $F_B$, we can always construct a continuous function $\tilde{F}_B$ such that $\mathrm{OPT}\text{-}\mathrm{W}(F_S,\tilde{F}_B)$ can be arbitrarily close to $\mathrm{OPT}\text{-}\mathrm{W}(F_S,F_B)$ and
$$\sup_{p\in \mathbb{R}_{+}}\mathrm{W}(p,F_S,F_B)=\sup_{p\in \mathbb{R}_{+}}\mathrm{W}(p,F_S,\tilde{F}_B)=\sup_{p\in \mathbb{R}_{+}}\mathrm{W}_1(p,F_S,\tilde{F}_B).$$
 Therefore, we have
\begin{equation*}
    \inf_{F_B,F_S\in\Delta_{L^1}(\mathbb{R}_{+})}\sup_{p\in\mathbb{R}_{+}}\frac{\mathrm{W}(p;F_S,F_B)}{\mathrm{OPT}\text{-}\mathrm{W}(F_S,F_B)}
    =\inf_{F_B,F_S\in\Delta_{L^1}(\mathbb{R}_{+})}\sup_{p\in\mathbb{R}_{+}}\frac{\mathrm{W}_1(p;F_S,F_B)}{\mathrm{OPT}\text{-}\mathrm{W}(F_S,F_B)}.
\end{equation*}

\section{Welfare Approximation in the General Case}
In this section, we consider the general case. We first introduce the optimization problem to facilitate determining the approximation ratios. Then we will show some useful properties of our problem in Section~\ref{sec:proper}. Using these properties, we characterize the extreme case where the tight bound can be achieved in Section~\ref{sec:bound}. Finally we propose an algorithm to approximate the optimal solution and give numerical bound results in Section~\ref{sec:numerical}.

We note that
$$\sup_{p\in\mathbb{R}_+}\frac{\textrm{W}_1(p;F_S,F_B)}{\textrm{OPT-W}(F_S,F_B)}=\sup_{F_P\in\Delta_{L^1}(\mathbb{R}_+)}\frac{\mathbb{E}_p[\textrm{W}_1(p;F_S,F_B)]}{\textrm{OPT-W}(F_S,F_B)}.$$
Furthermore, given a number $\beta$, proving that $\beta$ is a lower bound of the approximation ratio is equivalent to proving that the following formula is always non-negative:
\begin{align*}
    \inf_{F_B,F_S \in\Delta_{L^1}(\mathbb{R}_+)}\sup_{F_P\in\Delta_{L^1}(\mathbb{R}_+)}\left\{\mathbb{E}_p[\textrm{W}_1(p;F_S,F_B)] - \beta\cdot\textrm{OPT-W}(F_S,F_B)\right\}.
\end{align*}
Using the specific form in Section~\ref{sec:preliminaries}, we have
\begin{align*}
    &\mathbb{E}_p[\textrm{W}_1(p;F_S,F_B)] - \beta\cdot\textrm{OPT-W}(F_S,F_B)\\
    =&\mathbb{E}_p\left[\mathbb{E}_s\left[s+\left(\int_p^{\infty}(1-F_B(b))\mathrm{d}b+(1-F_B(p))(p-s)\right)\cdot\mathbf{1}_{p\ge s}\right]\right]- \beta\cdot\mathbb{E}_s\left[s+\int^{\infty}_s (1-F_B(b))\mathrm{d}b\right]\\
    =&\mathbb{E}_{p,s}\left[(1-\beta)s+
    \left(\int_p^{\infty}(1-F_B(b))\mathrm{d}b+(1-F_B(p))(p-s)\right)\cdot\mathbf{1}_{p\ge s}-\beta\int^{\infty}_s (1-F_B(b))\mathrm{d}b\right].
\end{align*}

For convenience, we set $H_B(b):=1-F_B(b)$ and $G_B(b):=\int_b^\infty H_B(p)\mathrm{d}p$. We have
\begin{align*}
    &\mathbb{E}_p[\textrm{W}_1(p;F_S,F_B)] - \beta\cdot\textrm{OPT-W}(F_S,F_B)\\
    =&\mathbb{E}_{p,s}\left[(1-\beta)s+
    \left(G_B(p)+H_B(p)(p-s)\right)\cdot\mathbf{1}_{p\ge s}-\beta G_B(s)\right].
\end{align*}
We also use $\Gamma(\beta,F_B,F_S,F_P)$ to denote $$\mathbb{E}_{p,s}\left[(1-\beta)s+
    \left(G_B(p)+H_B(p)(p-s)\right)\cdot\mathbf{1}_{p\ge s}-\beta G_B(s)\right].$$
    
Fixing $F_B$ and switching the order between $\inf_{F_S\in\Delta_{L^1}(\mathbb{R}_+)}$ and $\sup_{F_P\in\Delta_{L^1}(\mathbb{R}_+)}$, we have
\begin{equation}
     \inf_{F_S\in\Delta_{L^1}(\mathbb{R}_+)}\sup_{F_P\in\Delta_{L^1}(\mathbb{R}_+)}\Gamma(\beta,F_B,F_S,F_P)
    \ge\sup_{F_P\in\Delta_{L^1}(\mathbb{R}_+)}\inf_{F_S\in\Delta_{L^1}(\mathbb{R}_+)}\Gamma(\beta,F_B,F_S,F_P).
\label{ineq:infsup}
\end{equation}

The LHS and RHS of the above inequality correspond to the following two problems:
\begin{itemize}
    \item Given distributions $F_S$ and $F_B$, at least how much can fixed-price mechanisms achieve the optimal welfare?
    \item Given the distribution $F_B$, for any randomized fixed-price mechanism that uses only the knowledge of the buyer distribution, at least how much can it achieve of the optimal welfare for any seller distribution?
\end{itemize}
Based on Inequality~\eqref{ineq:infsup}, if $\beta$ is a lower bound of the solution of the second problem, it also applies to the first problem. Additionally, we prove that approximation ratios of these two problems are indeed equal in Section~\ref{sec:bound}. From now on, we are concerned with the second problem and deal with the RHS of Inequality~\eqref{ineq:infsup}. To prove $\beta$ is a lower bound, it is equivalent to show that given any $F_B\in\Delta_{L^1}(\mathbb{R}_+)$, there exists a price distribution $F_P$ such that the following inequality holds for any $s\in \mathbb{R}_+$:
\begin{align}
(1-\beta)s+\mathbb{E}_{p\sim F_P}[
    \left(G_B(p)+H_B(p)(p-s)\right)\cdot\mathbf{1}_{p\ge s}]-\beta G_B(s) \geq 0.
\label{formula:coef}
\end{align}

For a weakly increasing function $Q_P:\mathbb{R}_+\rightarrow \mathbb{R}_+$ with $Q_P(0)\ge 0$, we introduce $\phi(s,H_B,Q_P)$ to denote the expression
\begin{equation*}
 (1-\beta)s+\int_s^{\infty}
    \left(G_B(p)+H_B(p)(p-s)\right)\mathrm{d}Q_P(p)-\beta G_B(s).
\end{equation*}
Given the buyer distribution $F_B$, we consider the following optimization problem:
\begin{align}
\inf &~ \lim_{p\to\infty}Q_P(p) \notag\\
\text{s.t.}&~
\label{mec:opt}
\phi(s,H_B,Q_P)\geq 0, \forall s\in\mathbb{R}_+, \\
&~Q_P(p_1)\leq Q_P(p_2), \forall p_1,p_2\in\mathbb{R}_+, p_1< p_2. \notag
\end{align}

We prove that the optimal solution actually exists in Theorem~\ref{thm:properties}. Let $Q^*_P$ denote an optimal solution that achieves the optimal objective $\inf \lim_{p\to\infty}Q_P(p)$.
To prove $\beta$ is a lower bound of the approximation ratio, it is equivalent to show that $\lim_{p\to\infty}Q_P^*(p)\le 1$. Here is the idea. If $\lim_{p\to\infty}Q_P^*(p)\le 1$, we set $F_P(0)= Q_P^*(0)-\lim_{p\to\infty}Q_P^*(p) +1$ and $F_P(p)=Q_P^*(p)-Q_P^*(0)+F_P(0)$ for any $p\in \mathbb{R}_+$. It indicates that $\lim_{p\to\infty}F_P(p)=1$. The function $F_P$ is a cumulative distribution function and satisfies Inequality~\eqref{formula:coef}. If $\lim_{p\to\infty}Q_P^*(p) > 1$, there is no way to find a feasible cumulative distribution function $F_P$.

\subsection{The Properties of an Optimal Solution}
\label{sec:proper}
We present some characterizations of an optimal solution $Q_P^*$. Recall that we focus on how to design $Q_P$ to minimize the objective function while satisfying all constraints in Problem~\eqref{mec:opt}. We consider the coefficient concerned with $\mathrm{d}Q_P(p)$ and let $\psi(p,s)=G_B(p)+H_B(p)(p-s)$ and have that 
$$\frac{\partial \psi(p,s)}{\partial p}=H_B'(p)(p-s)\le 0.$$
That is, the function $\psi(p,s)$ weakly decreases in terms of variable $p$. We also know that $\psi(p,s)\ge 0$ for all $p\in[s,\infty)$. For any $p_1<p_2$ and $s$, $\mathrm{d}Q_P(p_1)$ has a weakly larger positive weight than $\mathrm{d}Q_P(p_2)$.
The idea is that we treat $\mathrm{d}Q_P(p)$ as resource and it is better to put resource on the locations with larger weight than that with smaller weight. 

We define $s_2$ as the unique solution to the equation $(1-\beta)s_2=\beta G_B(s_2)$. We provide the form of an optimal solution directly.
\begin{theorem}\label{thm:properties}
Given $F_B$, an optimal solution $Q_P^*$ satisfies that
\begin{itemize}
\item if $s\le s_2$, the closed form of its derivative function $q_P^*(s)$ is
     $$\beta \left(\frac{H_B(s)}{G_B(s)}-\frac{\int_s^{s_2} H_B(t)^2\mathrm{d}t}{G_B(s)^2}\right)+\left(1-\beta\right)\cdot\frac{G_B(s_2)}{G_B(s)^2};$$
\item if $s>s_2$, the equation $Q_P^*(s)=Q_P^*(s_2)$ holds.
\end{itemize}
\end{theorem}

Note that $Q_P^*$ is continuous and it is differentiable in $[0, s_2]$. We first show that $Q_P^*$ satisfies constraints of Problem~\eqref{mec:opt}.
\begin{proposition}\label{prop:opt}
The function $Q_P^*$ in Theorem~\ref{thm:properties} satisfies:
\begin{itemize}
\item For any $s\le s_2$, it holds that $\phi(s,H_B,Q_P^*)=0$.
\item For any $s>s_2$, it holds that $\phi(s,H_B,Q_P^*)>0$.
\item For any $p_1,p_2\in\mathbb{R}_+$ and $p_1< p_2$, it holds that $Q_P^*(p_1)\le Q_P^*(p_2)$.
\end{itemize}
\end{proposition}

\begin{proof}
Given any $s\le s_2$, we rewrite $\int_s^{s_2} H_B(p)q_P^*(p)\mathrm{d}p$ as:
\begin{align*}
&\int_s^{s_2}H_B(p)\left[\beta \left(\frac{H_B(p)}{G_B(p)}-\frac{\int_p^{s_2} H_B(t)^2\mathrm{d}t}{G_B(p)^2}\right)+\left(1-\beta\right)\cdot\frac{G_B(s_2)}{G_B(p)^2}\right]\mathrm{d}p\\
=&\frac{1}{G_B(s)}\cdot \left[\beta \int_s^{s_2} H_B(t)^2\mathrm{d}t+(1-\beta)\int_s^{s_2} H_B(t)\mathrm{d}t\right].
\end{align*}
Therefore, we get
$$G_B(s)\int_s^{s_2} H_B(p)q_P^*(p)\mathrm{d}p
    =\beta \int_s^{s_2} H_B(t)^2\mathrm{d}t+(1-\beta)\int_s^{s_2} H_B(t)\mathrm{d}t.$$
Next, we calculate the derivative of the above formula about $s$ and get
\begin{equation*}
H_B(s)\int_s^{s_2} H_B(p)q_P^*(p)\mathrm{d}p+G_B(s)H_B(s)q_P^*(s)
=\beta H_B(s)^2+(1-\beta) H_B(s).
\end{equation*}
We know that $G_B(s_2)=\frac{(1-\beta)s_2}{\beta}>0$. It indicates that $H_B(s)>0$ for any $s\le s_2$. Thus, we divide every term by $H_B(s)$ and get 
$$(1-\beta)-\int_s^{s_2} H_B(p)q_P^*(p)\mathrm{d}p\\
    -G_B(s)q_P^*(s)+\beta H_B(s)=0.$$
Notice that the equation $Q_P^*(s)=Q_P^*(s_2)$ holds if $s>s_2$. It means that $\int_{s_2}^{\infty} H_B(p)\mathrm{d}Q_P^*(p)=0$. Therefore, we plug $\int_{s_2}^{\infty} H_B(p)\mathrm{d}Q_P^*(p)$ into the above equation and get
$$(1-\beta)-\int_s^{\infty} H_B(p)\mathrm{d}Q_P^*(p)
    -G_B(s)q_P^*(s)+\beta H_B(s)=0.$$
We consider the integral. Note that $(1-\beta)s_2=\beta G_B(s_2)$. We have
$$(1-\beta)s+\int_{s}^{\infty}[G_B(p)+H_B(p)(p-s)]\mathrm{d}Q_P^*(p)-\beta G_B(s)=0.$$
We finish the proof of the first statement that $\phi(s,H_B,Q_P^*)=0$ for any $s\le s_2$.

For the second statement, we have
$\phi(s,H_B,Q_P^*)=(1-\beta)s-\beta G_B(s)$ for any $s> s_2$ since $Q_P^*(s)=Q_P^*(s_2)$. Considering the derivative about $s$ and we have $\phi'(s,H_B,Q_P^*)=1-\beta+\beta H_B(s)>0$. We know that $\phi(s_2,H_B,Q_P^*)=0$. Therefore, we finish the proof that $\phi(s,H_B,Q_P^*)>0$ for any $s>s_2$.

For the third statement, if $s\le s_2$, the derivative function $q_P^*(s)>0$. And if $s>s_2$, the equation $Q_P^*(s)=Q_P^*(s_2)$ holds. The function $Q_P^*$ is non-decreasing. Thus, it holds that $Q_P^*(p_1)\le Q_P^*(p_2)$ for any $p_1,p_2\in\mathbb{R}_+$ and $p_1< p_2$.

To summarize, the function $Q_P^*$ satisfies allconstraints of Problem~\eqref{mec:opt}.\qedhere
\end{proof}

Next, we present that for any $\hat{Q}_P$ satisfying constraints of Problem ~\eqref{mec:opt}, it holds that $\lim_{p\to\infty}\hat{Q}_P(p)\ge \lim_{p\to\infty}Q_P^*(p)$. 
Therefore, $Q_P^*$ is an optimal solution.
\begin{proof}[Proof of Theorem ~\ref{thm:properties}]
From Proposition~\ref{prop:opt}, we notice that $\phi(s,H_B,Q_P^*)=0$ for any $s\le s_2$. Therefore, for any $\hat{Q}_P$ satisfying constraints of Problem ~\eqref{mec:opt}, we have
\begin{equation}\label{ineq:xi}
\phi(s,H_B,Q_P^*)\le\phi(s,H_B,\hat{Q}_P), \forall s\le s_2.
\end{equation}

Let $\theta(s):=G_B(0)\cdot\left[-H_B(s)\int_0^s\frac1{G_B(t)^2}\mathrm{d}t+\frac1{G_B(s)}\right].$ Note that $\theta'(s)=-G_B(0)H_B'(s)\int_0^s\frac1{G_B(t)^2}\mathrm{d}t\ge 0$. We define the function $$\xi(H_B,Q_P):=\int_0^{\infty}\psi(p,0)\mathrm{d}Q_P(p)+\int_0^{s_2}\int_s^{\infty}\psi(p,s)\mathrm{d}Q_P(p)\mathrm{d}\theta(s).$$

We know that $\xi(H_B,Q_P^*)\le \xi(H_B,\hat{Q}_P)$ by Inequality~\eqref{ineq:xi}. Next, we rewrite $\xi(H_B,Q_P)$ as:
\begin{align*}
&\int_0^{s_2}\theta(p)G_B(p)\mathrm{d}Q_P(p)+\int_0^{s_2}H_B(p)\int_0^p\theta(s)\mathrm{d}s\mathrm{d}Q_P(p)\\
+&\theta(s_2)\int_{s_2}^{\infty}\psi(p,s_2)\mathrm{d}Q_P(p)+\int_{s_2}^{\infty}H_B(p)\int_0^{s_2}\theta(s)\mathrm{d}s\mathrm{d}Q_P(p).
\end{align*}
The function $\xi(H_B,Q_P)$ can be represented by $\int_0^{s_2}\tau_1(p)\mathrm{d}Q_P(p)+\int_{s_2}^{\infty}\tau_2(p)\mathrm{d}Q_P(p)$. Note that $\int_0^p\theta(s)\mathrm{d}s=G_B(0)G_B(p)\int_0^p\frac1{G_B(t)^2}\mathrm{d}t$, hence we have that
$$\tau_1(p)=\theta(p)G_B(p)+H_B(p)\int_0^p\theta(s)\mathrm{d}s=G_B(0).$$

As for $\tau_2(p)$, we have $$\tau_2(p)=\theta(s_2)\psi(p,s_2)+H_B(p)\int_0^{s_2}\theta(s)\mathrm{d}s.$$ The derivative $\tau_2'(p)=H_B'(p)[\theta(s_2)(p-s_2)+\int_0^{s_2}\theta(s)\mathrm{d}s]\le0$ for any $p\ge s_2$. Thus, we have $\tau_2(p)\le \tau_2(s_2)=G_B(0)$. 

Consequently, we have
\begin{align*}
&\int_0^{\infty}G_B(0)\mathrm{d}\hat{Q}_P(p)\\
\ge&\int_0^{s_2}G_B(0)\mathrm{d}\hat{Q}_P(p)+\int_{s_2}^{\infty}\tau_2(p)\mathrm{d}\hat{Q}_P(p)\\
\ge&\int_0^{s_2}G_B(0)\mathrm{d}Q_P^*(p)+\int_{s_2}^{\infty}\tau_2(p)\mathrm{d}Q_P^*(p)\\
=&\int_0^{\infty}G_B(0)\mathrm{d}Q_P^*(p).
\end{align*}
The second step is because that $\xi(H_B,Q_P^*)\le \xi(H_B,\hat{Q}_P)$.  The last step is because that $Q_P^*(s)=Q_P^*(s_2)$ for any $s>s_2$. Finally, we have $\lim_{p\to\infty}\hat{Q}_P(p)\ge \lim_{p\to\infty}Q_P^*(p)$.
\end{proof}

\subsection{The Bound of the Approximation Ratio}
\label{sec:bound}
Recall that $\beta$ is an lower bound if and only if $\lim_{p\to\infty}Q_P^*(p)\le 1$.
Considering properties of $Q_P^*$ in Theorem~\ref{thm:properties}, we use the result below to figure out the lower bound of the exact approximation ratio.
\begin{theorem}
Given $\beta>0$, $\beta$ is a lower bound of the approximation ratio if and only if for any $F_B\in\Delta_{L^1}(\mathbb{R_+})$, it satisfies that
\begin{equation*}
\int^{s_2}_0\left[\beta \left(\frac{H_B(s)}{G_B(s)}-\frac{\int^{s_2}_s H_B(t)^2\mathrm{d}t}{G_B(s)^2}\right)+\left(1-\beta\right)\frac{G_B(s_2)}{G_B(s)^2}\right]\mathrm{d}s\le 1.
\end{equation*}
\label{thm:approx}
\end{theorem}

Since both $G_B(s)$ and $\int^{s_2}_s H_B(t)^2\mathrm{d}t$ can be expressed by the function $H_B$. Therefore, the objective above is only determined by $H_B$. It is easy to see that the function $H_B$ is scale free in terms of the variable $s$, that is, $\hat{H}_B(s)=H_B(\sigma s)$ for $\sigma>0$ achieving the same objective. WLOG we assume that $s_2=1$ and search all possible functions $H_B\in\Delta_{L^1}[0,1]$ which are weakly decreasing to maximize the objective $\mathrm{obj}(H_B)$:
\begin{equation*}
\int^1_0\left[\beta \left(\frac{H_B(s)}{G_B(s)}-\frac{\int^1_s H_B(t)^2\mathrm{d}t}{G_B(s)^2}\right)+\left(1-\beta\right)\frac{G_B(1)}{G_B(s)^2}\right]\mathrm{d}s,
\end{equation*}
where $(1-\beta)\cdot1=\beta\cdot G_B(1)$. That is $G_B(1)=\frac{1-\beta}{\beta}$.

We know that $1\ge H_B(0)\ge H_B(1)> 0$, in which $H_B(1)>0$ is due to $G_B(1)>0$. However, we allow that $H_B(1)$ equals zero by considering the extreme case that there is a very large number with a very low probability in the buyer distribution. 
Assuming that $H_B^*$ is the optimal solution and setting $z_1:=\sup\{z|H_B^*(z)=1\}$ and $z_2:=\inf\{z|H_B^*(z)=0\}$, we present the characterization of $H_B^*$.
\begin{theorem}
The optimal function $H_B^*$ achieving the maximum of the objective $\mathrm{obj}(H_B)$ satisfies that, for any  $z\in [z_1, z_2]$, we have
       \begin{equation*}
{H_B^*}'(z)=\frac{\int^{z_2}_z H_B^*(t)^2\mathrm{d}t-H_B^*(z)G_B^*(z)-G_B(1)^2}{G_B^*(z)^3\int^z_0 \frac1{G_B^*(t)^2}\mathrm{d}t}.  
\end{equation*}
\end{theorem}

\begin{proof}
To the begin with, we relax the constraint that $H_B^*$ is weakly decreasing on the interval $[z_1,z_2]$. Therefore, we treat the optimization problem as a variational problem. We use $H_B(z)=H_B^*(z)+\epsilon\eta(z)$ to denote an arbitrary function defined on $[z_1,z_2]$. Note that $H_B\in\Delta_{L^1}[z_1,z_2]$ does not affect the integral $\int_{z_2}^1 q_P^*(z)\mathrm{d}z$. Thus, we define
\begin{equation*}
I(H_B)=\int^{z_2}_0\left[\beta \left(\frac{H_B(z)}{G_B(z)}-\frac{\int^1_z H_B(t)^2\mathrm{d}t}{G_B(z)^2}\right)+\left(1-\beta\right)\frac{G_B(1)}{G_B(z)^2}\right]\mathrm{d}z.
\end{equation*}
Considering $\frac{\partial I}{\partial\epsilon}\bigg|_{\epsilon=0}$, we have
\begin{align*}
\int^{z_2}_{z_1}\bigg[&\frac{\beta}{G_B^*(z)}-\int^z_0 \frac{\beta H_B^*(s)}{G_B^*(s)^2}\mathrm{d}s-\int^z_0 \frac{2\beta H_B^*(z)}{G_B^*(s)^2}\mathrm{d}s\\
-&\int^z_0 \left(-\beta\int^{z_2}_s H_B^*(t)^2\mathrm{d}t+(1-\beta)G_B(1)\right)\cdot\frac2{G_B^*(s)^3}\mathrm{d}s\bigg]\eta(z)\mathrm{d}z.
\end{align*}

By the first order condition, we have $\frac{\partial I}{\partial\epsilon}\bigg|_{\epsilon=0}=0$. Notice that $\eta(z)$ can be arbitrary. Therefore, the coefficient of $\eta(z)$ should be always zero, so we have that
\begin{equation*}
\frac1{G_B^*(z)}=\int^z_0 \left[\frac{H_B^*(s)}{G_B^*(s)^2}+\frac{2H_B^*(z)}{G_B^*(s)^2}+\frac{2G_B(1)^2-2\int^{z_2}_s H_B^*(t)^2\mathrm{d}t}{G_B^*(s)^3}\right]\mathrm{d}s.
\end{equation*}

We calculate the derivative and have
 \begin{equation*}
{H_B^*}'(z)=\frac{\int^{z_2}_z H_B^*(t)^2\mathrm{d}t-H_B^*(z)G_B^*(z)-G_B(1)^2}{G_B^*(z)^3\int^z_0 \frac1{G_B^*(t)^2}\mathrm{d}t}.  
\end{equation*}

Next, we demonstrate that ${H_B^*}'(z)<0$, that is, the function $H_B^*$ satisfies the decreasing constraint naturally. It implies that $H_B^*$ is actually the optimal solution that achieves the maximum of the objective $\mathrm{obj}(H_B)$. We prove that ${H_B^*}'(z)<0$ by contradiction. Suppose that $z_3=\sup\{z|{H_B^*}'(z)\ge0\}$. It indicates that ${H_B^*}'(z)<0$ for any $z\in(z_3,z_2]$. For the numerator, given a sufficient small $\delta>0$, we consider
\begin{align*}
&\int^{z_2}_{z_3} H_B^*(t)^2\mathrm{d}t-H_B^*(z_3)G_B^*(z_3)-G_B(1)^2,\\
&\int^{z_2}_{z_3+\delta} H_B^*(t)^2\mathrm{d}t-H_B^*(z_3+\delta)G_B^*(z_3+\delta)-G_B(1)^2.
\end{align*}
We know that $H_B^*(z_3)=H_B^*(z_3+\delta)-\delta{H_B^*}'(z_3+\delta)$. Take the difference of above formulas into account.  We have
\begin{align*}
&\delta{H_B^*}'(z_3+\delta)G_B^*(z_3)-\int^{z_3+\delta}_{z_3} H_B^*(t)\left(H_B^*(t)-H_B^*(z_3+\delta)\right)\mathrm{d}t\\
<&\delta{H_B^*}'(z_3+\delta)G_B^*(z_3)-\delta^2 H_B^*(z_3){H_B^*}'(z_3+\delta)\\
=&\delta{H_B^*}'(z_3+\delta)\left(G_B^*(z_3)-\delta H_B^*(z_3)\right)
\end{align*}
Note that ${H_B^*}'(z_3+\delta)<0$. Thus, the difference is less than 0.  We derive that ${H_B^*}'(z_3)<0$, which contradicts with the assumption that ${H_B^*}'(z_3)\ge0$.
\end{proof}

\begin{remark}
Although we get the structure of $H_B^*$, it is hard to figure out the expression of $H_B^*$ explicitly. Therefore, we choose the numerical method to get the lower bound later.
\end{remark}

In addition, suppose that $\hat{\beta}$ is the approximation ratio only using the knowledge of buyer distribution. 
We prove that $\hat{\beta}$ is also the approximation ratio with full knowledge of buyer and seller distributions. 
We have the following theorem:
\begin{theorem}\label{thm:same beta}
The randomized fixed-price mechanism that  only uses the knowledge of $F_B$ has the same approximation ratio as the fixed-price mechanism that uses both $F_B$ and $F_S$. Additionally, for the worst pair $(\hat{F}_B,\hat{F}_S)$, we have
$$\hat{F}_S(p)=G_B(1)\left[-\hat{H}_B(p)\int_0^p\frac1{\hat{G}_B(s)^2}\mathrm{d}s+\frac1{\hat{G}_B(p)}\right].$$
\end{theorem}
It means that given $\hat{F}_B$, the following equality holds
\begin{equation*}
\inf_{F_S\in\Delta_{L^1}(\mathbb{R}_+)}\sup_{F_P\in\Delta_{L^1}(\mathbb{R}_+)}\Gamma(\hat{\beta},\hat{F}_B,F_S,F_P)
=\sup_{F_P\in\Delta_{L^1}(\mathbb{R}_+)}\inf_{F_S\in\Delta_{L^1}(\mathbb{R}_+)}\Gamma(\hat{\beta},\hat{F}_B,F_S,F_P).
\end{equation*}
The proof idea is that we can find a seller's distribution such that $\textrm{W}_1(p;F_S,\hat{F}_B) - \hat{\beta}\cdot\textrm{OPT-W}(F_S,\hat{F}_B)$ is identical for any price $p\in[0,1]$. And the seller's distribution is exactly $\hat{F}_S$. 

\begin{proof}
First, we show that gain from trade $\textrm{GFT}(p;\hat{F}_S,\hat{F}_B)$ is identical for all $p\in[0,1]$. We have
\begin{align*}    
\hat{F}_S(p)=&\hat{G}_B(1)\left[-\hat{H}_B(p)\int_0^p\frac1{\hat{G}_B(s)^2}\mathrm{d}s+\frac1{\hat{G}_B(p)}\right]\\
=&\left(\hat{G}_B(p)\int_0^p\frac{\hat{G}_B(1)}{\hat{G}_B(s)^2}\mathrm{d}s\right)'.
\end{align*}

Thus, we have
$$\frac{\int_0^p \hat{F}_S(s)\mathrm{d}s}{\hat{G}_B(p)}=\int_0^p\frac{\hat{G}_B(1)}{\hat{G}_B(s)^2}\mathrm{d}s.$$

We calculate the derivative and get
\[\hat{F}_S(p)\hat{G}_B(p)+\hat{H}_B(p)\int_0^p \hat{F}_S(s)\mathrm{d}s=\hat{G}_B(1).\]

We rewrite $\textrm{GFT}(p;\hat{F}_S,\hat{F}_B)$ as:
\begin{align*}
&\mathbb{E}_{B,S}[(b-s)\cdot\mathbf{1}_{b\geq p \geq s}] \\
    =&\Pr(s\le p)\mathbb{E}_{B}[(b-p)\cdot\mathbf{1}_{b\geq p}]+\Pr(b\ge p)\mathbb{E}_S[(p-s)\cdot\mathbf{1}_{p\ge s}]\\
    =&\hat{F}_S(p)\int_p^{\infty}\left(1-\hat{F}_B(b)\right)\mathrm{d}b+(1-\hat{F}_B(p))\int_0^p \hat{F}_S(s)\mathrm{d}s\\    =&\hat{F}_S(p)\hat{G}_B(p)+\hat{H}_B(p)\int_0^p \hat{F}_S(s)\mathrm{d}s.
\end{align*}

Therefore, for the worst pair $(\hat{F}_S,\hat{F}_B)$, gain from trade is a constant for any $p\in[0,1]$. Furthermore, the following equality holds:
\begin{equation}\label{eq:FS_hat}
    \textrm{W}_1(p;\hat{F}_S,\hat{F}_B) - \hat{\beta}\cdot\textrm{OPT-W}(\hat{F}_S,\hat{F}_B)=c.
\end{equation}

Suppose that $\hat{F}_P$ is the optimal price mechanism. We have
\[\mathbb{E}_{p\sim\hat{F}_P}[\textrm{W}_1(p;\hat{F}_S,\hat{F}_B)] - \hat{\beta}\cdot\textrm{OPT-W}(\hat{F}_S,\hat{F}_B)=0.\]
It indicates that $c=0$. We also have
\begin{align*}
    &\max_{p\in[0,1]} \left\{\textrm{W}_1(p;\hat{F}_S,\hat{F}_B) - \hat{\beta}\cdot\textrm{OPT-W}(\hat{F}_S,\hat{F}_B)\right\}\\
    \ge& \inf_{F_S\in\Delta_{L^1}(\mathbb{R}_+)}\max_{p\in[0,1]} \left\{\textrm{W}_1(p;F_S,\hat{F}_B) - \hat{\beta}\cdot\textrm{OPT-W}(F_S,\hat{F}_B)\right\}\\
    =&\inf_{F_S\in\Delta_{L^1}(\mathbb{R}_+)}\sup_{F_P\in\Delta_{L^1}[0,1]} \Gamma(\hat{\beta},\hat{F}_B,F_S,F_P)\\
    \ge&\sup_{F_P\in\Delta_{L^1}[0,1]}\inf_{F_S\in\Delta_{L^1}(\mathbb{R}_+)} \Gamma(\hat{\beta},\hat{F}_B,F_S,F_P)\\
    =& \inf_{F_S\in\Delta_{L^1}(\mathbb{R}_+)} \left\{\mathbb{E}_{p\sim \hat{F}_P}[\textrm{W}_1(p;F_S,\hat{F}_B)] - \hat{\beta}\cdot\textrm{OPT-W}(F_S,\hat{F}_B)\right\}
\end{align*}

By Equation~\eqref{eq:FS_hat}, the first formula should be zero. Since $\hat{F}_P$ is the optimal price mechanism, the last formula is also zero. Consequently, we have
\begin{equation*}
    \inf_{F_S\in\Delta_{L^1}(\mathbb{R}_+)}\sup_{F_P\in\Delta_{L^1}[0,1]} \Gamma(\hat{\beta},\hat{F}_B,F_S,F_P)
    =\sup_{F_P\in\Delta_{L^1}[0,1]}\inf_{F_S\in\Delta_{L^1}(\mathbb{R}_+)} \Gamma(\hat{\beta},\hat{F}_B,F_S,F_P)\qedhere
\end{equation*}
\end{proof}

\subsection{Numerical results}\label{sec:numerical}
Now we try to approach the lower bound by the methods of dicretization and dynamic programming. We show that there exists a step function $\hat{H}_B$ that approximates the objective of $H_B^*$, where $\hat{H}_B$ weakly decreases at points over $\left\{0,\frac1{n},\ldots,\frac{n}{n}\right\}$. Then, we can search the optimal step function $\hat{H}_B$ with $\epsilon$ granularity for any $s\in\left\{0,\frac1{n},\ldots,\frac{n}{n}\right\}$ to maximize $\mathrm{obj}(\hat{H}_B)$.  Given $\beta$, we define $M:=\max \mathrm{obj}(\hat{H}_B)$. After adding the discretization error to $M$, if the result is less than $1$, the corresponding $\beta$ could be a lower bound. As for the search algorithm in discrete grids, we use dynamic programming (see Algorithm~\ref{alg:dp}). We first present the framework of the algorithm and discuss the concrete discretization error later. We denote by $Sol(x,y,z,s_k)$ the objective value of the following problem:
\begin{align*}
\max& \int^{s_k}_0\left[\beta \left(\frac{\hat{H}_B(s)}{\hat{G}_B(s)}-\frac{\int^1_s \hat{H}_B(t)^2\mathrm{d}t}{\hat{G}_B(s)^2}\right)+\left(1-\beta\right)\frac{G_B(1)}{\hat{G}_B(s)^2}\right]\mathrm{d}s\\
\text{s.t.}~&  \int_{s_k}^1 \hat{H}_B(t)^2\mathrm{d}t=x,\\
&\int_{s_k}^1 \hat{H}_B(t)\mathrm{d}t=y,\\
&\hat{H}_B(s_k)=z.
\end{align*}
Note that $s_k$ is in the set of $\left\{0,\frac1n,\ldots,1\right\}$.
Then, we show the original framework of the algorithm as follows.
\begin{algorithm}[!htbp]
\caption{Framework of dynamic programming}
\KwIn{$\beta,n,\epsilon$}
\KwOut{$M$}
\BlankLine
set $s_0=0$ and $k=0$\;
Initialize $Sol(x,y,z,s_k)$ by an $\frac{n}{\epsilon^2}\times\frac{n}{\epsilon}\times\frac1{\epsilon}$ tensor with zeros\;
\While{$s_k<1$}{
    $s_{k+1}\leftarrow s_k+\frac1n$\;
    Initialize $Sol(x,y,z,s_{k+1})$ by an $\frac{n}{\epsilon^2}\times\frac{n}{\epsilon}\times\frac1{\epsilon}$ tensor with zeros\;
    \ForEach{cell $\in Sol(x,y,z,s_k)$}{
        \For{$z'=0: z$}{
        $x'\leftarrow x-z'\cdot z'/n$\;
        $y'\leftarrow y-z'/n$\;
        \If{$x'\ge 0$ and $y'\ge 0$ and $Sol(x,y,z,s_k)+\int_{s_k}^{s_{k+1}} q_P(s)\mathrm{d}s>Sol(x',y',z',s_{k+1})$}{
            $Sol(x',y',z',s_{k+1})\leftarrow Sol(x,y,z,s_k)+\int_{s_k}^{s_{k+1}} q_P(s)\mathrm{d}s$\;
            }
        }
    }
    $k\leftarrow k+1$\;
}
$M\leftarrow \max_{cell\in Sol(x,y,z,1)} cell$\;
\label{alg:dp}
\end{algorithm}

Intuitively, the discretization error depends on both $\frac1n$ and $\epsilon$. The complexity is $O\left(\frac{n^3}{\epsilon^5}\right)$.

To reduce the time complexity, we first use the following two inequalities to eliminate some special cases directly:
\begin{align*}
(1-s)\int_s^1 H_B(t)^2 \mathrm{d}t &\ge \left(\int_s^1 H_B(t) \mathrm{d}t\right)^2,\\
H_B(s)\int_s^1 H_B(t) \mathrm{d}t &\ge \int_s^1 H_B(t)^2 \mathrm{d}t.
\end{align*}

Next, we construct a new function $K(\cdot)$ with $\frac{\epsilon}n$ granularity to approximate $\int_{s_k}^1 \hat{H}_B(t)^2\mathrm{d}t$. The time complexity can be reduced by $\frac1{\epsilon}$ times. We construct functions on grid points as follows:
\begin{align*}
\hat{H}_B\left(\frac{i}n\right)&=\left\lfloor H_B^*\left(\frac{i}n\right)\cdot\frac1{\epsilon}\right\rfloor\cdot\epsilon;\\
K\left(\frac{i}n\right)&=\sum_{j=i+1}^n \left\lfloor H_B^*\left(\frac{j}n\right)^2\cdot\frac1{\epsilon}\right\rfloor\cdot\frac{\epsilon}n;\\
\hat{G}_B\left(\frac{i}n\right)&=\sum_{j=i+1}^n \hat{H}_B\left(\frac{j}n\right)\cdot\frac1n+G_B(1).
\end{align*}
For $s\in\left(\frac{i}n, \frac{i+1}n\right)$, we use the following continuous extensions:
\begin{align*}
\hat{H}_B(s)&=\hat{H}_B\left(\frac{i+1}n\right),\\
K(s)&=K\left(\frac{i+1}n\right)+\left(\frac{i+1}n-s\right)\cdot\hat{H}_B\left(\frac{i+1}n\right)^2,\\
\hat{G}_B(s)&=\hat{G}_B\left(\frac{i}n\right)+\left(\frac{i+1}n-s\right)\cdot\hat{H}_B\left(\frac{i+1}n\right).
\end{align*}

Then, we define $\mathrm{obj}(\hat{H}_B,K)$ as:
$$\mathrm{obj}(\hat{H}_B,K)=\int^1_0\left[\beta \left(\frac{\hat{H}_B(s)}{\hat{G}_B(s)}-\frac{K(s)}{\hat{G}_B(s)^2}\right)+\left(1-\beta\right)\frac{G_B(1)}{\hat{G}_B(s)^2}\right]\mathrm{d}s.$$

We use $\mathrm{obj}(\hat{H}_B,K)$ to approximate $\mathrm{obj}(H_B^*)$. Next, we analyze the discretization error. 
\begin{theorem}
For the discretization error, we have $\mathrm{obj}(H_B^*)-\mathrm{obj}(\hat{H}_B,K)\le \beta\ln \frac{1-\beta}{1-\beta\left(1+\epsilon+\frac1n\right)}$.
\label{thm:step}
\end{theorem}
\begin{proof}
First, we have
\begin{align*}
    & \mathrm{obj}(H_B^*)-\mathrm{obj}(\hat{H}_B,K)\\
    =& \underbrace{\beta\int_0^1\left(\frac{H_B^*(s)}{G_B^*(s)}-\frac{\hat{H}_B(s)}{\hat{G}_B(s)}\right)\mathrm{d}s}_{\gamma_1}+\underbrace{\beta\int_0^1 \left(\frac{K(s)}{\hat{G}_B(s)^2}-\frac{\int^{1}_{s}H_B^*(t)^2\mathrm{d}t}{G_B^*(s)^2}\right)\mathrm{d}s}_{\gamma_2}\\
    &+\underbrace{(1-\beta)G_B(1)\int_0^1 \left(\frac1{G_B^*(s)^2}-\frac1{\hat{G}_B(s)^2}\right)\mathrm{d}s}_{\gamma_3}.
\end{align*}

Now we will upper bound three parts $\gamma_1,\gamma_2$ and $\gamma_3$ in above equation respectively.

Note that functions $H_B^*$ and $\hat{H}_B$ are both decreasing functions. Therefore, we have
\begin{align*}
    &G_B^*(s)-\hat{G}_B(s)\\
    \le& H_B^*(s)\cdot\left(\frac{i+1}n-s\right)+\sum_{j=i+1}^{n-1} H_B^*\left(\frac{j}n\right)\cdot\frac1n-\hat{H}_B\left(\frac{i+1}n\right)\cdot\left(\frac{i+1}n-s\right)-\sum_{j=i+2}^n \hat{H}_B\left(\frac{j}n\right)\cdot\frac1n\\
    \le&\left(H_B^*(s)-\left(H_B^*\left(\frac{i+1}n\right)-\epsilon\right)\right)\cdot\left(\frac{i+1}n-s\right)+\sum_{j=i+1}^{n-1}\left(H_B^*\left(\frac{j}n\right)-\left(H_B^*\left(\frac{j}n\right)-\epsilon\right)\right)\cdot\frac1n\\
    \le&\frac1n+\epsilon.
\end{align*}

For any $s\in\left(\frac{i}n,\frac{i+1}n\right)$, we have $\hat{H}_B(s)\le H_B^*(s)$. That is, we always have $\hat{G}_B(s)\le G_B^*(s)$ and $K(s)\le \int^1_s H_B^*(t)^2\mathrm{d}t$. Therefore, we get $\gamma_3 \le 0$ directly.

We get that $\gamma_2$ is upper bounded by
\begin{align*}
    &\beta\int_0^1 \left(\int^1_s H_B^*(t)^2\mathrm{d}t\right)\cdot\left({\hat{G}_B(s)^{-2}}-{G_B^*(s)^{-2}}\right)\mathrm{d}s\\
    \le&\beta\int_0^1 \left(H_B^*(s)\int^1_s H_B^*(t)\mathrm{d}t\right)\cdot\left({\left(\int^1_s H_B^*(t)\mathrm{d}t+G_B(1)-\epsilon-\frac1n\right)^{-2}}-{\left(\int^1_s H_B^*(t)\mathrm{d}t+G_B(1)\right)^{-2}}\right)\mathrm{d}s.
\end{align*}
Let $x:=\int^1_s H_B^*(t)\mathrm{d}t$ and $a:=G_B(1)-\epsilon-\frac1n$. We get that $\gamma_2$ is upper bounded by
\begin{align*}
    &\beta\int_0^{G_B^*(0)-G_B(1)} \left(\frac{x}{(a+x)^2}-\frac{x}{(G_B(1)+x)^2}\right)\mathrm{d}x\\
    =&\beta\left(\ln(a+x)+\frac{a}{a+x}-\ln\left(G_B(1)+x\right)-\frac{G_B(1)}{G_B(1)+x}\right)\bigg|_0^{G_B^*(0)-G_B(1)}.
\end{align*}

For $\gamma_1$, we have
\begin{align*}
    \gamma_1=&\beta\left(\ln\hat{G}_B(s)\bigg|_0^1-\ln G_B^*(s)\bigg|_0^1\right)\\
    =&\beta\left(\ln G_B^*(0)-\ln \hat{G}_B(0)\right)\\
    \le&\beta\left(\ln G_B^*(0)-\ln \left(G_B^*(0)-\frac1n-\epsilon\right)\right).
\end{align*}
Hence, the part $\gamma_1+\gamma_2$ is upper bounded by
$$\beta\left(\ln G_B(1)-\ln a+\frac{a}{a+G_B^*(0)-G_B(1)}-\frac{G_B(1)}{G_B^*(0)}\right)\le\beta\ln \frac{G_B(1)}{a}.$$
To summarize, the discretization error is $\beta\ln \frac{1-\beta}{1-\beta\left(1+\epsilon+\frac1n\right)}$.\qedhere
\end{proof}

Consequently, we update $\tilde{Sol}(x,y,z,s_k)$ as follows:
\begin{align*}
\max &\int^{s_k}_0\left[\beta \left(\frac{\hat{H}_B(s)}{\hat{G}_B(s)}-\frac{K(s)}{\hat{G}_B(s)^2}\right)+\left(1-\beta\right)\frac{G_B(1)}{\hat{G}_B(s)^2}\right]\mathrm{d}s\\
\text{s.t.}~&
K(s_k)=x,\\
&\int_{s_k}^1 \hat{H}_B(t)\mathrm{d}t=y,\\
&\hat{H}_B(s_k)=z.
\end{align*}

Therefore, we can run our algorithm with $O\left(\frac{n^3}{\epsilon^4}\right)$ complexity. Note that $x'\leftarrow x-\left\lfloor \frac{z'\cdot z'}{\epsilon}\cdot\frac{\epsilon}n\right\rfloor$ in every step. We present the calculation of $\int_{s_k}^{s_{k+1}} q_P(s)\mathrm{d}s$:
\begin{align*}
&\int_{s_k}^{s_{k+1}}\frac{\beta\left(\hat{H}_B(s_{k+1})\hat{G}_B(s_{k+1})-K(s_{k+1})\right)+(1-\beta)G_B(1)}{(\hat{G}_B(s_{k+1})+(s_{k+1}-s)\hat{H}_B(s_{k+1}))^2}\mathrm{d}s\\
    =&\frac{\beta\left(\hat{H}_B(s_{k+1})\hat{G}_B(s_{k+1})-K(s_{k+1})\right)+(1-\beta)G_B(1)}{\hat{G}_B(s_{k+1})(\hat{G}_B(s_{k+1})+(s_{k+1}-s_k)\hat{H}_B(s_{k+1}))}\cdot(s_{k+1}-s_k).
\end{align*}

Finally, we are fully prepared to solve the lower bound $\beta$.
\begin{theorem}
Given distribution $F_B$, for  randomized fixed-price mechanism that uses only the knowledge of the buyer distribution, the designer can achieve at least $0.71$ of the optimal welfare for any seller distribution $F_S$. Given distributions $F_S$ and $F_B$, the designer can always select a price that achieves at least $0.71$ of the optimal welfare.
\end{theorem}
\begin{proof}
We implement the dynamic programming algorithm by Python 3.9, with parameters $\beta=0.71$, $n=75$ and $\epsilon=\frac1{75}$. We run our algorithm with a computing cloud node whose resource contains the Intel Gold 5218 and 128GB memory. Our program takes 552794s (about 6d). Finally, we get $M=0.9440$. Besides, the discretization error is $0.0479$ by Theorem~\ref{thm:step}. Thus, we have that $\mathrm{obj}(H_B^*)\le 0.9919<1$. That is, we find a lower bound $\beta=0.71$. All codes can be found at our \href{https://github.com/ruc-renzy/Improved-Approximation-Ratios-of-Fixed-Price-Mechanisms-in-Bilateral-Trades}{github repository}.
\end{proof}
We also summarize lower bounds for different discretization factors in Table~\ref{tab:lb}.

\begin{table}[!htbp]
\centering
\begin{tabular}{|c|c|c|c|c|c|c|}\hline
     $\beta$& $n$ & $\epsilon$ & $M$ & error & $\mathrm{obj}(H_B^*)$ & running time \\ \hline
     $0.69$& $35$ & $\frac1{35}$ & $0.9056$ & $0.0939$ & $0.9995$ & $4166$s ($\approx 1$h) \\\hline
     $0.7$& $50$ & $\frac1{50}$ & $0.9245$ & $0.0686$ & $0.9931$ & $42708$s ($\approx 12$h) \\\hline
      $0.71$& $75$ & $\frac1{75}$ & $0.9440$ & $0.0479$ & $0.9919$ & $552794$s ($\approx 6$d) \\\hline
\end{tabular}
\caption{A summary of different factors}
\label{tab:lb}
\end{table}
In the end, we choose $\beta=0.7381$ and construct an instance function $H_B$. We show that $\int^{1}_0 q_P^*(s)\mathrm{d}s>1$. Consequently, we get an upper bound.
\begin{theorem}
\label{thm:upperbound}
Given distribution $F_B$, the randomized fixed-price mechanism that uses only the knowledge of the buyer distribution has approximation ratio at most $0.7381$. Given distributions $F_S$ and $F_B$, no fixed-price mechanism has
an approximation ratio better than $0.7381$.
\end{theorem}
\begin{proof}
Let $\beta:=0.7381$ and consider the following function on $[0,1]$:
 \begin{equation*}
    H_B(s)=
    \begin{cases}
    1,& s\le0.25\\
    1-\lambda\left(1-e^{1.5-6s}\right),& 0.25< s< 0.6\\
    0,& 0.6\le s\le 1\\
    \end{cases}
\end{equation*}
where $\lambda=1/(1-e^{-2.1})$. Besides, the function $H_B$ satisfies $G_B(1)=\frac{1-\beta}{\beta}$. Then, we calculate the integral $\int^{1}_0 q_P^*(s)\mathrm{d}s$, that is
$$\int^{1}_0\beta \left[\left(\frac{H_B(s)}{G_B(s)}-\frac{\int^{1}_s H_B(t)^2\mathrm{d}t}{G_B(s)^2}\right)+(1-\beta)\frac{G_B(1)}{G_B(s)^2}\right]\mathrm{d}s.$$
We divide the integral into three parts: $\int^1_{0.6} q_P^*(s)\mathrm{d}s$, $\int^{0.6}_{0.25} q_P^*(s)\mathrm{d}s$ and $\int^{0.25}_0 q_P^*(s)\mathrm{d}s$.

For the first part, we have
\begin{align*}
    \int^1_{0.6} q_P^*(s)\mathrm{d}s = \int^1_{0.6} (1-\beta)\frac1{G_B(1)}\mathrm{d}s = 0.4\beta=0.29524.
\end{align*}

For the second part, it is difficult to compute accurately. So we use a program to approximate it and have $\int^{0.6}_{0.25} Q_P^*(s)\mathrm{d}s\approx0.41766$.

For the third part, the integral $\int^{0.25}_0 q_P^*(s)\mathrm{d}s$ is
\begin{align*}
     & \int^{0.25}_0 \frac{\beta\left(G_B(0.25)-\int^{1}_{0.25} H_B(t)^2\mathrm{d}t\right)+(1-\beta)G_B(1)}{\left(0.25-s+G_B(0.25)\right)^2}\mathrm{d}s\\
    =&\frac{A}{0.25-s+G_B(0.25)}\bigg|_0^{0.25}\\
    =&\frac{0.25A}{0.25G_B(0.25)+G_B(0.25)^2}.
\end{align*}
where $A=\beta\left(G_B(0.25)-\int^{1}_{0.25} H_B(t)^2\mathrm{d}t\right)+(1-\beta)G_B(1)$. For each term, we compute it precisely and have $\int^{0.25}_0 q_P^*(s)\mathrm{d}s\approx 0.28722$.

Finally, we get $\int^{1}_0 q_P^*(s)\mathrm{d}s\approx1.00012>1$. That is, given the function $H_B$ above, no matter how we design the price mechanism, there always exists a seller distribution $F_S$ such that the welfare ratio is less than $0.7381$. Please check numerical results at our github repository.
\end{proof}

\section{Single-sample Approximation}
We consider the fixed-price mechanism that receives a single sample from the seller distribution as sole prior information in this section. Previous works~\cite{Dutting:2021wx,Kang:2022wo} study  deterministic mechanisms.  We extend their bounds to randomized mechanisms. For asymmetric case, \cite{Dutting:2021wx} show that every deterministic IC, IR and BB mechanism has approximation ratio at most $1/2$. Here we show that the randomization would not help. In other words, given any single sample from the seller distribution, the designer is allowed to use a randomization over different fixed-price mechanisms. We demonstrate that a randomized fixed-price mechanism has  approximation ratio at most $1/2$. \cite{Dutting:2021wx} also prove that there exists a deterministic fixed-price mechanism of which the lower bound is $1/2$. Therefore, for randomized fixed-price mechanisms, approximation ratio $1/2$ is also tight.

\begin{theorem}
Every randomized fixed-price mechanism that receives a single sample
from the seller’s distribution as information has approximation ratio at most $1/2$.
\label{thm:asy_sample_up}
\end{theorem}
The idea is that we fix the buyer's value to be quite high. At the same time, we construct the seller distribution on many small discrete values. Under such construction, the approximation ratio is exactly the probability that the trade occurs. That is, for any randomized mechanism, we show that there exists an instance that the seller rejects the trade with probability $0.5$.
\begin{proof}
For a randomized fixed-price mechanism, we denote it as a function $F_P: \mathbb{R}\rightarrow \Delta_{L^1}(\mathbb{R})$ that maps a value to a distribution. Then, we construct a set $\{s_1, s_2,\ldots,s_n\}$ that any two adjacent values satisfy $\Pr\left[F_P(s_k)< s_{k+1}\right]=1-\epsilon$ and $s_k\le s_{k+1}$ for $k\in [n-1]$. Consider a seller whose valuation is drawn uniformly from the set and a buyer whose valuation equals $s_n\cdot 2^n$. Let $s'$ denote the seller's sample and $s$ denote the seller's valuation. Then the probability that the seller rejects the trade is:
\begin{align*}
    \Pr_{s',s}\left[F_P(s')< s\right] &= \sum_{i=1}^n \Pr\left[F_P(s)<
                                      s_i\right]\cdot \Pr \left[ s=s_i \right] \\
    &\geq \sum_{i=2}^n \frac{i-1}{n}\cdot(1-\epsilon)\cdot\frac1{n}\\
    &=\frac{n-1}{2n}\cdot(1-\epsilon).
\end{align*}

As $n$ grows to infinity and $\epsilon$ decreases to $0$, the probability that the seller rejects the trade approaches $1/2$. The approximation ratio approaches the probability that the trade occurs. Therefore, the upper bound of any randomized mechanism is $1/2$.
\end{proof}

As for symmetric case, \cite{Kang:2022wo} show that the approximation ratio $3/4$ is tight recently. We extend their bound to the randomized mechanisms.
\begin{theorem}
\label{thm:randomized}
Every randomized fixed-price mechanism that receives a single sample
from the distribution as information in the symmetric bilateral trade has approximation ratio at most $3/4$.
\end{theorem}

\begin{proof}
Similar to the proof of Theorem~\ref{thm:asy_sample_up}, for any randomized
mechanism, we want to show that there exists a distribution that achieves $3/4$
of welfare. Suppose that the randomized mechanism $F_P(\cdot)$ maps the single
sample value to a price (note that it is also a random variable). Given $F_P(\cdot)$
and any $\epsilon>0$,
we construct a sequence of positive numbers $\hat{s}=\{s_1, s_2,\ldots, s_n\}$ such that
$\Pr\left[F_P(s_i)<s_{k+1}\right]\ge 1-\epsilon$, where $i=1,2,\ldots,k$. Besides,
we introduce a value $s_a$ which is high enough (say $s_a\ge s_n^2$). The support of the value
distribution is exactly $\{s_1, s_2,\ldots, s_n, s_a\}$ in which the probability
that the value is any of $\{s_1,s_2,\ldots,s
_n\}$ is identical, i.e., $\Pr \left[ s=s_i \right]=\Pr \left[ s=s_j \right]$
where $i,j\in [n]$. We denote by $q$ the total probability that the value is from the sequence
$\hat{s}$, that is $\Pr \left[ s\in \hat{s} \right]=q$, and $\Pr \left[ v=s_a \right]=1-q$. For convenience, we
denote such a distribution as $F$.

For the optimal welfare, we consider two situations. The first one is that at
least one of the seller and buyer has the value $s_a$. The other is that
neither seller nor buyer has the value $s_a$. In this case, if the value is less
than $s_n$, we increase it to $s_n$, this can only improve the welfare. Thus, we have
$$\textrm{OPT-W}(F)\leq s_a\cdot(1-q^2)+s_n\cdot q^2.$$

Given the distribution $F$ and the randomized mechanism $F_P(\cdot)$, we denote
by $\mathrm{LOSS}(F,F_P)$ the difference between the optimal welfare and the
expected welfare by the randomized mechanism mentioned above. We only consider
the case that the buyer's value is $s_a$ and  the seller's value is greater than
the price $F_P(\cdot)$. We can lower bound $\mathrm{LOSS}(F,F_P)$, where $\bar{s}$ is the
value sample from the distribution $F$:
\begin{align*}
   &\sum_{i=1}^n (s_a-s_i)\cdot \Pr\left[F_P(\bar{s})< s_i\right]\cdot \Pr\left(s=s_i\right)\cdot \Pr\left(b=s_a\right)\\
    \geq&(s_a-s_n)\cdot(1-q)\cdot(1-\epsilon)\cdot\frac{q}{n}\cdot\sum_{i=2}^n \frac{(i-1)q}{n}\\
    =&(s_a-s_n)\cdot(1-q)\cdot(1-\epsilon)\cdot\frac{(n-1)q^2}{2n}.
\end{align*}
The second inequality is due to that we only consider that $\bar{s}=s_1,s_2,\ldots,s_{i-1}$.
Now we compute the ratio between the welfare loss and the optimal welfare. Note that $s_a$ is very high, we have
\begin{align*}
    \frac{\mathrm{LOSS}(F,F_P)}{\textrm{OPT-W}(F)}\geq&\frac{s_a\cdot(1-q)\cdot(1-\epsilon)\cdot\frac{(n-1)q^2}{2n}}{s_a\cdot(1-q^2)}\\
    =&\frac{n-1}{2n}\cdot(1-\epsilon)\cdot\frac{q^2}{1+q}.
\end{align*}

For the term $\frac{q^2}{1+q}$, we know that it always increases if $q\in[0,1]$. Consequently, we have $\mathrm{LOSS}(F,F_P)\ge\frac14\cdot\textrm{OPT-W}(F)$ when $\epsilon\to 0$, $n\to\infty$ and $q\to 1$. That is, the upper bound of any randomized mechanism is $3/4$.
\end{proof}

\section{Conclusion and Discussion}
We improved approximation ratios of fixed-price mechanisms in several settings for the bilateral trade problem with one buyer and one seller, in terms of the optimal welfare. Our methods differs from previous work. We establish the lower bounds by solving an optimization problem. Then we try to find a nearly-optimal solution using the discretization method and dynamic programming. Ideally, we can bound the error caused by discretization within any precision. However, due to the limitation of the hardware, we still leave a gap in approximation ratios.

Several interesting questions arise. Note that we only use the knowledge of the buyer distribution and get Theorem~\ref{thm:approx}. In fact, we can also start from the seller distribution instead and get a similar but not the same optimization problem given in Theorem~\ref{thm:approx}. If we only use the knowledge of the seller distribution, would we have a worse approximation bound? Another challenging question is to come up with a simple mechanism to approximate the optimal gain-from-trade. We do also think the ratio of RandOff~\cite{Deng:2022uy,Fei:2022uh} could be improved significantly, but it needs new and deep insights.

\bibliographystyle{alpha}
\bibliography{ref}
\end{document}